\documentclass[runningheads]{llncs}
\usepackage[T1]{fontenc}
\usepackage{graphicx}
\usepackage{xspace}
\usepackage{lmodern}
\usepackage{hhline}
\usepackage{todonotes}
\usepackage{microtype}
\usepackage{amssymb}
\usepackage{makecell}
\newcommand\Comment[1]{}
\usepackage{comment}
\usepackage{amsmath}
\usepackage{thmtools} 
\usepackage{thm-restate}
\usepackage{mathtools}

\usepackage{tikz}
\usetikzlibrary{matrix, calc, fit}

\DeclareTextFontCommand{\texttt}{\ttfamily\upshape}

\DeclareMathOperator*{\argmin}{arg\,min}
\DeclareMathOperator*{\pre}{pre}
\DeclareMathOperator*{\Tf}{\mathbf{T}}
\newcommand{\cO}{\mathcal{O}}
\newcommand{\FPT}{{\sf FPT}\xspace}
\newcommand{\NP}{{\sf NP}\xspace}
\newcommand{\Pclass}{{\sf P}\xspace}
\newcommand{\str}[1]{\texttt{#1}}
\def\dd{\mathinner{.\,.}}

\newcommand{\Scal}{\mathcal S}

\newcommand{\dSH}{\ensuremath{\partial_{\str{SH}}}}
\newcommand{\dS}{\ensuremath{\partial_{\str{S}}}}
\newcommand{\dH}{\ensuremath{\partial_{\str{Ham}}}}

\newcommand{\anyConsensus}[1]{\text{\normalfont\str{(#1)}}-\allowbreak\text{{\sc Consensus}}\xspace}

\newcommand{\rHConsensusM}{\anyConsensus{r,\dH}-\textsc{Mixed}}
\newcommand{\rsHConsensusM}{\anyConsensus{rs,\dH}-\textsc{Mixed}}
\newcommand{\rsSHconsensus}{\anyConsensus{rs,\dSH}}
\newcommand{\rSHconsensus}{\anyConsensus{r,\dSH}}
\newcommand{\sSHconsensus}{\anyConsensus{s,\dSH}}
\newcommand{\rsSconsensus}{\anyConsensus{rs,\dS}}
\newcommand{\rSconsensus}{\anyConsensus{r,\dS}}
\newcommand{\sSconsensus}{\anyConsensus{s,\dS}}
\newcommand{\permGroup}{\ensuremath{\mathfrak{S}_n}\xspace}
\newcommand{\perm}[1]{\ensuremath{\mathfrak{S}_{#1}}\xspace}
\newcommand{\swapstr}[1]{\ensuremath{h_{#1}}\xspace}

\newcommand{\rsDistconsensus}{\anyConsensus{rs,$\partial$}}
\newcommand{\rDistconsensus}{\anyConsensus{r,$\partial$}}
\newcommand{\sDistconsensus}{\anyConsensus{s,$\partial$}}
\newcommand{\probdef}[3]{
\vspace{2mm}
\noindent\fbox{
   \begin{minipage}{0.96\textwidth}
   \textsc{#1}\\
   {\bf{Input:}} #2  \\
   {\bf{Question:}} #3
   \end{minipage}
   }
   \vspace{2mm}
}

\begin{document}
\title{String Consensus Problems\\ with Swaps and Substitutions}
\author{Estéban Gabory\inst{1}\orcidID{0000-0002-9897-1512}\and Laurent Bulteau\inst{2}\orcidID{0000-0003-1645-9345} \and
Gabriele Fici\inst{1}\orcidID{0000-0002-3536-327X} \and
 Hilde Verbeek \inst{3}\orcidID{0000-0002-2399-3098} }
\authorrunning{E. Gabory et al.}
\institute{Dipartimento di Matematica e Informatica, Università di Palermo, Palermo, Italy \and LIGM, CNRS, Université Gustave Eiffel, F77454 Marne-la-vallée, France \and
CWI, Amsterdam, The Netherlands}
\maketitle              %
\begin{abstract}
String consensus problems aim at finding a string that minimizes some given distance with respect to an input set of strings. In particular, in the \textsc{Closest String} problem, we are given a set of strings of equal length and a radius~$d$. The goal is to find a new string that differs from each input string by at most~$d$ substitutions. We study a generalization of this problem where, in addition to substitutions, swaps of adjacent characters are also permitted, each operation incurring a unit cost. Amir et al.\ showed that this generalized problem is \NP-hard, even when only swaps are allowed. In this paper, we show that it is \FPT with respect to the parameter~$d$. Moreover, we investigate a variant in which the goal is to minimize the \emph{sum} of distances from the output string to all input strings. For this version, we present a polynomial-time algorithm.


\end{abstract}

\section{Introduction}
Consensus problems aim at finding, given a collection of objects and a distance function, an object that is as close as possible to each and every object in the collection. In the {\sc Closest String} problem, objects are strings and the distance constraint is given by an integer $d$: the solution must be at Hamming distance at most $d$ (i.e., at most $d$ substitutions away) from each input string. We call such a constraint a \emph{radius constraint}, which is generally harder to satisfy than the corresponding \emph{sum constraint}: minimize the sum of distances (indeed,  {\sc Closest String} is \NP-hard~\cite{FRA} and \FPT for parameter~$d$~\cite{GRA}, whereas the sum of distances can easily be optimized by taking the most frequent character at each position).
Betzler et al.~\cite{betzler2011average} give an algorithmic framework to show that the median problem (i.e., the sum-only variant) for various distances is \FPT for the average distance between input objects. %

{\sc Closest String} has many application domains. In bioinformatics, where strings such as genomes from the same species are often compared~\cite{CHEN2012164}, the use of the radius~$d$ as a parameter is justified by the similarity of the inputs.  A possible extension of {\sc Closest String} consists of choosing, among all solution strings satisfying the radius constraint, the one minimizing the sum of distances to input strings~\cite{amir-3-strings}. Incidentally, most algorithms for variants of \emph{radius-only} {\sc Closest String} can be adapted to this \emph{radius and sum} problem (in particular, {\sc Closest String} is \FPT for $d$ in this setting as well~\cite{BUL}).

In this paper, we consider not only substitution operations, but also \emph{swaps}, i.e., exchange of consecutive characters, and define the \emph{Swap distance} (\dS). 
Swaps are used as an extension to the Levenshtein distance (yielding the Damerau–Levenshtein distance), e.g., in spell checkers where they help improve predictions \cite{damerau}. 
Importantly, their formalism enforces that swaps are pairwise disjoint, hence some pairs of strings may be incomparable even if they use the same multi-set of characters\footnote{The variant where swaps may overlap corresponds to the Kendall-tau distance.}. 
For example, strings \str{abcd} and \str{badc} are at distance 2, but \str{abc} and \str{bca} are incomparable, since swaps are pairwise disjoint. 
We further consider the \emph{Swap+Hamming} distance (written \dSH), where permitted operations are both substitutions and swaps (on distinct characters). 

The aim of this paper is to study the parameterized complexity of {\sc Closest String} for both Swap and Swap+Hamming measures, and for any of the radius-only, sum-only, or radius+sum objective functions. 
The main parameters we consider are the radius bound $d$ and the alphabet size. %
For both distances, we show that the radius-only variant is \FPT, and that the sum-only variant is in \Pclass. Additionally, for the Swap distance, we show that the radius+sum variant is \FPT for $d$.

\paragraph*{Problems.}

We consider the following three problems, where $\partial$ is any distance function.%

\probdef{\rDistconsensus}{A set of $k$ strings $\Scal=\{s_1, \dd , s_k\}$ each of length $n$, integer~$d$.
}{Is there a length-$n$ string~$s^*$ such that~$\displaystyle \max_{s\in \Scal} \partial(s,s^*)\leq d$?}

\probdef{\sDistconsensus}{A set of $k$ strings $\Scal=\{s_1, \dd , s_k\}$ each of length $n$, integer~$D$.
}{Is there a length-$n$ string~$s^*$ such that~$\displaystyle \sum_{s\in \Scal} \partial(s,s^*)\leq D$?}

\probdef{\rsDistconsensus}{A set of $k$ strings $\Scal=\{s_1, \dd , s_k\}$ each of length $n$, integers~$d$ and $D$.
}{Is there a length-$n$ string~$s^*$ such that~$\displaystyle \max_{s\in \Scal} \partial(s,s^*)\leq d$ and $\displaystyle \sum_{s\in \Scal} \partial(s,s^*)\leq D$?}

Note that previously-studied problems can be named within this formalism, e.g., {\sc Closest String} is \anyConsensus{r,\dH} and \anyConsensus{r,Levenshtein} is known as {\sc Center String}. %

\paragraph*{Our contributions.}

We summarize our main contributions in Table~\ref{table:our-contributions}.

 Additionally, we prove that neither $\rSconsensus(d,n,|\Sigma|)$ nor $\rsSconsensus(d,n,|\Sigma|)$ (Theorem~\ref{thm:results swap}), nor $\rSHconsensus(d,n,|\Sigma|)$, nor $\rsSHconsensus(d,n,|\Sigma|)$ (Theorem~\ref{thm:SH no polyk}) admit a polynomial kernel unless $\NP\subseteq \text{co}\NP/\text{poly}$.
\begin{table}
    \centering
    \begin{tabular}{c||c@{\ }>{\small}c|c@{\ }>{\small}c}                                    
                      & \dS& Section~\ref{sec:ds} & \dSH & Section~\ref{sec:dsh}\\
\hhline{=::====}        \makecell{Sum \\(\normalfont\str{s})}             &  $\cO(kn)$ & Cor.~\ref{cor:s-swap}    &   $\cO(k^3n^2)$ &Thm.~\ref{thm:sSHconsensus poly} %
\\
\hhline{-||----}        \makecell{Radius\\ (\normalfont\str{r})}     &  \makecell{$\FPT(d)$  \\ Kernel of size $\cO(k^2\log d)$ }    &  \makecell{  Thm.~\ref{thm:results swap}\\ Thm.~\ref{thm:results swap}\\ }& \makecell{$\FPT(d)$ %
}& \makecell{Thm.~\ref{thm:rSHconsensus FPT} %
} \\ 
\hhline{-||----}    \makecell{Radius and Sum \\(\normalfont\str{rs})}   & \makecell{$\FPT(d)$  \\ Kernel of size $\cO(k^2\log d)$ \\ Kernel of size $\cO(D^3\log D)$}  & \makecell{Thm.~\ref{thm:results swap}\\  Thm.~\ref{thm:results swap}\\ Thm.~\ref{thm:results swap}\\  }      &  \makecell{Open%
} & \makecell{%
} 
    \end{tabular}\\[2mm]
    \label{table:our-contributions}
    \caption{Summary of our contributions.}
\end{table}

\paragraph*{Related work.}

Consensus problems under the Hamming distance have been extensively studied from both classical and parameterized complexity perspectives. It is known that $\anyConsensus{r,\dH}$ is \NP-complete~\cite{FRA}. However, when parameterized by the radius $d$, the problem becomes fixed-parameter tractable (FPT), and $\anyConsensus{r,\dH}(d,k)$ admits a polynomial kernel of size $\mathcal{O}(k^2 d \log k)$~\cite{GRA}. In the sum-of-distances variant $\anyConsensus{s,\dH}$, the problem becomes much simpler: the optimal consensus string is the column-wise majority string, and can be computed in linear time~\cite{GRA}. The branching algorithm underlying the FPT result for $\anyConsensus{r,\dH}$ can be extended to handle the combined radius-and-sum variant $\anyConsensus{rs,\dH}$, which also lies in $\FPT(d)$~\cite{BUL}. With respect to kernelization, while no polynomial kernel exists for $\anyConsensus{r,\dH}$ and $\anyConsensus{rs,\dH}$ when parameterized by $(d, n)$ unless $\NP \subseteq \text{co}\NP/\text{poly}$, both $\anyConsensus{r,\dH}(d,k)$ and $\anyConsensus{rs,\dH}(d,k)$ have polynomial kernels of size $\mathcal{O}(k^2 d \log k)$, and $\anyConsensus{rs,\dH}$ additionally admits a polynomial kernel of size $\mathcal{O}(D^3 \log D)$~\cite{BasavarajuEtAl2018,BUL}. The problem $\sSconsensus$ was introduced in~\cite{AMI}, which showed that \sSconsensus is \NP hard. Before that, pattern matching with swaps was introduced in~\cite{10.1007/3-540-60044-2_50}, and solved in~\cite{AMIR2000247}. Note that other definitions of swap distance exist, e.g., in~\cite{cunha2025complexityalgorithmsswapmedian}, where swap distance is considered between permutations and allows for swapping several times at a given position (namely, the Kendall-tau distance). 

\section{Preliminaries}

\paragraph*{Parameterized Complexity.}
We use general notions of parameterized complexity and kernels \cite{Downey,FominEtAl2019Book}. In a nutshell (omitting some technical difficulties not irrelevant to our case), for a parameterized problem $P(k)$, an \emph{FPT algorithm} is an exact algorithm solving $P$ in $f(k)\mathrm{poly}(n)$ time where $n$ is the instance size and $k$ is the parameter. A \emph{kernel} is a polynomial-time algorithm reducing an instance of $P(k)$ to another instance with size $f(k)$. Such a kernel is \emph{polynomial} if $f$ is polynomial. 
It is known that there exists a kernel for $P(k)$ if and only if $P(k)\in \FPT$, but having a polynomial kernel is seemingly more restrictive. Finally, the \emph{parameterized reductions} we use in this paper are polynomial-time reductions that also preserve the value of the parameter (up to some polynomial). In other words, a parameterized reduction from $P(k)$ to $Q(k')$ is a function $f$ with: \textbf{(i)} $(x, k)$ is a yes-instance for $P$ if and only if $(x',k')=f(x,k)$ is a yes-instance for $Q$, and \textbf{(ii)} $k' \le \mathrm{poly}(k)$. With such a parameterized reduction, any \FPT (resp.~polynomial kernel) for $Q(k')$ can be carried over as an \FPT algorithm (resp.~polynomial kernel) for $P(k)$ \cite{BodlaenderEtAl2011}.

\paragraph*{Strings.}
Let $\Sigma^n$ be the set of \emph{strings} $s=s[1]\cdots s[n]$ of \emph{length} $|s|=n$ over an alphabet $\Sigma$, the elements of which are called \textit{letters}. Given a letter $\str{a}\in\Sigma$, we write $|s|_{\str{a}}=|\{p\in[1 \dd n]~|~s[p]=\str{a}\}$|. The \emph{reversal} of a string $s$ is the string $s^r=s[|s|]\cdots s[|1|]$. %
Given an arbitrary interval $I=[i \dd j]\subseteq [1 \dd n]$, we write $s[I]=s[i \dd j]$ for the \emph{substring} of $s$ starting at position $i$ and ending at position $j$, and we write $s[j \dd i]=s[i \dd j]^r$. A \emph{prefix} of $s$ is a substring of the form $s[1 \dd p]$ and a \emph{suffix} of $s$ is a substring of the form $s[p \dd |s|]$, for some $p$. By $s_1\cdot s_2$ we denote the \emph{concatenation} of two strings $s_1$ and $s_2$, i.e., $s_1\cdot s_2=s_1[1]\cdots s_1[|s_1|]\cdot s_2[1]\cdots s_2[|s_2|]$. We say that two strings $s_1$, $s_2$ are \emph{matching} at a position $p$ if $s_1[p]=s_2[p]$, and that they are \emph{mismatching} %
otherwise. We say that $s_2$ is obtained from $s_1$ by a \emph{substitution} at position $p$ with letter $\str{a}$ if $s_1[p]\neq s_2[p]=\str{a}$ and $s_1[p']=s_2[p']$ when $p'\neq p$. Given a set $\Scal=\{s_1, \dd , s_k\}$ containing strings having the same length, we call \emph{$p$th column} the set $\Scal[p]=\{s[p]|~s\in\Scal\}$. We say that $\Scal[p]$ is \emph{dirty} if it contains at least two distinct elements; otherwise, we say that it is \emph{clean}. We also write $\Scal[i\dd j]$ for the \emph{segment} of $\Scal$ between $i$ and $j$, which is defined as the set $\{s[i\dd j],s\in\Scal\}$. Given two binary strings $h_1$ and $h_2$, we write $h_1\oplus h_2$ for the string obtained by taking the position-wise \emph{XOR} of both input strings.

\paragraph*{Swaps.}
Given an integer $n$, we denote by \permGroup the set of all permutations of $[n]=[1 \dd n]$.
Given an integer $p$, the \emph{swap at positions $(p,p+1)$} (or simply at position $p$) is the operation transforming a string $s$ into $s[1]\cdots s[p-1] \cdot  s[p+1] \cdot s[p] \cdot s[p+2] \cdots s[|s|]$. Two swaps at positions $p$ and $q$ respectively are \emph{disjoint} or \emph{compatible} if $|p-q|\geq 2$ (in which case, the two operations commute).
A \emph{swap permutation} of size $m$ is a set of $m$ pairwise-compatible swaps (it can also be seen as a permutation of \permGroup consisting of $m$ support-disjoint adjacent transpositions). The \emph{swap string} $h_\sigma$ of a swap permutation $\sigma$ is the length-$(n-1)$ binary string $h$ such that $h[p]=\str{1}$ if $\sigma$ contains a swap at position $(p,p+1)$, and $h[p]=\str{0}$ otherwise.  

Two strings $s_1$ and $s_2$ are \emph{matching} if there exists a swap permutation transforming $s_1$ into $s_2$. To ensure the uniqueness of the swap permutations between strings, we say that such a permutation is \emph{valid} if it never swaps two identical letters: any swap permutation can be made valid by removing such superfluous swaps. In that case, we write \perm{s_1,s_2} for this permutation and \swapstr{s_1,s_2} for the corresponding swap string. We state the following lemma:

\begin{restatable}{lemma}{unicity}\label{lem:unicity}
Given two matching strings $s_1$ and $s_2$, \perm{s_1,s_2} and \swapstr{s_1,s_2} are unique and can be computed in linear time.
\end{restatable}
\begin{proof}
Let $\sigma$ and $\tau$ be two swap permutations between $s_1$ and $s_2$. Since permutations have a unique decomposition into support-disjoint cycles (here, transpositions) \cite{Dixon1996}, the string $h_{s_1,s_2}$ is well-defined and unique if and only if $\sigma=\tau$.  

If $\sigma\neq\tau$, let $(p,p+1)$ be the last positions that are swapped by exactly one of the permutations, without loss of generality, $\sigma$ swaps $(p,p+1)$ but not $\tau$. By definition, it implies that $s_2[p+1]=s_1[p]\neq s_1[p+1]=s_2[p]$. But since $s_2[p+1]=s_1[\tau(p+1)]$, we necessarily have $\tau(p+1)\neq p+1$, and since $(p,p+1)$ is, by assumption, not swapped by $\tau$, then $(p+1,p+2)$ has to be a swapped by $\tau$. But since $\sigma$ swaps $(p,p+1)$, that is incompatible with $(p+1,p+2)$, those positions are swapped by exactly one of the permutations. This contradicts our initial assumption.

We can compute $h_{s_1,s_2}$ in linear time by reading the string from left to right, and starting to swap at the first unmatching symbol. Indeed, if for some $p_0$ with $s_1[p_0]\neq s_2[p_0]$ we have $s_1[p]=s_2[p]$ for every $p\le p_0-1$, there cannot be a swap at positions $(p_0-1,p_0)$, so there is a swap at positions $(p_0,p_0+1)$ if the strings are matching. We conclude by induction.
\qed \end{proof}

We say that matching strings $s_1$ and $s_2$ have a \emph{swap} at position $(p,p+1)$ (or simply at position $p$) if $\perm{s_1,s_2}$ contains the corresponding swap (i.e.,  $\swapstr{s_1,s_2}[p]=\str{1}$), and we write $(p,p+1)\in\perm{s_1,s_2}$. Note that this is more precise than simply having $s_1[p]=s_2[p+1]$ and $s_2[p+1]=s_1[p]$, as can be seen in the following example.

\begin{example}\label{bad swap}
Let  $s_1=\str{abab}$ and $s_2=\str{baba}$.
One has $\swapstr{s_1,s_2}=\str{101}$ (so their swap distance is $2$), and $s_1$ and $s_2$ do not have a swap at position 2, even though $s_1[2,3]=\str{ba}$ and $s_2[2,3]=\str{ab}$. 
\end{example}

\begin{remark}\label{obs:invariant}
    Given two strings $s_1$, $s_2$ of length $n$ that are matching, one has $|s_1|_{x}=|s_2|_{x}$ for any letter $x\in\Sigma$. Furthermore, given a position $p$ such that $h_{s_1,s_2}[p]=\str{0}$, one has $|s_1[1 \dd p]|_{x}=|s_2[1 \dd p]|_{x}$, and $|s_1[p+1 \dd n]|_{x}=|s_2[p+1 \dd n]|_{x}$, for any letter $x\in\Sigma$. This can be verified by considering the swap permutation as two distinct ones, respectively on the prefixes and suffixes of $s_1$ and $s_2$.
\end{remark}

The following technical Lemma is illustrated in Figure~\ref{fig:three-way-match}. %

\begin{restatable}{lemma}{threewaymatch}\label{lem:three-way-match}
Given three strings $s_1,s_2,s_3$ such that both $s_1,s_2$ and $s_2,s_3$ are matching, let $h=\swapstr{s_1,s_2}\oplus \swapstr{s_2,s_3}$.
If $h$ does not have two consecutive $\str{1}$s, then $s_1$ and $s_3$ are matching and  $h=\swapstr{s_1,s_3}$.
If $h$ has two consecutive $\str{1}$s at positions $p-1$ and $p$ for some $p$ (i.e., $h[p-1]=h[p]=\str{1}$), then $s_1$ and $s_3$ do not match, and for any string $t$ matching both $s_1$ and $s_3$ we have $t[p-1 \dd p+1]=s_2[p-1 \dd p+1]$.
\end{restatable}%
\begin{proof}
For the first statement, one can observe that if $h$ does not have two consecutive $\str{1}$s, then by definition all swaps $(p,p+1)$ with $h[p]=\str{1}$ are compatible. Those swaps are exactly those that $s_1$ and $s_2$, or $s_2$ and $s_3$ have, but not both. Since swaps are involutions, we obtain $s_3$ after performing those swaps on $s_1$ (remembering that compatible swaps commute). We conclude that $h=h_{s_1,s_3}$ by uniqueness (Proposition~\ref{lem:unicity}).

For the second statement, let us assume that $h_{s_1,s_2}\oplus h_{s_2,s_3}[p-1,p]=\str{11}$ for some $p\in[2 \dd n-1]$. Without loss of generality, $(p-1,p)\in \perm{s_1,s_2}$ (hence $(p,p+1)\not\in \perm{s_1,s_2}$) and $(p,p+1)\in \perm{s_2,s_3}$ (hence $(p-1,p)\not\in\perm{s_2,s_3}$). Let us write $\str{x}= s_1[p-1]=s_2[p]=s_3[p+1]$, $\ell=|s_1[1 \dd p-1]|_{\str{x}}$ and $r=|s_1[p+1 \dd n]|_{\str{x}}$. Since $(p-1,p)\in\perm{s_1,s_2}$, we have $s_1[p]\neq s_1[p-1]=\str{x}$, so $|s_1|_{\str{x}}=|s_1[1 \dd p-1]|_{\str{x}}+|s_1[p+1 \dd n]|_{\str{x}}=\ell+r$. Now, since $(p,p+1)\not\in\perm{s_1,s_2}$, one has by Observation~\ref{obs:invariant} that 
 $|s_2[p+1 \dd n]|_{\str{x}}=|s_1[p+1 \dd n]|_{\str{x}}=r$, and since $s_2[p]=x$ one has $|s_2[p \dd n]|_{\str{x}}=r+1$. Finally, since $(p-1,p)\not\in\perm{s_2,s_3}$, we know that $|s_3[p \dd n]|_{\str{x}}=|s_2[p \dd n]|_{\str{x}}=r+1$. 

Now, assume that $s_1$ and $s_3$ are matching, namely that there exists a swap transformation $\perm{s_1,s_3}$. Since $s_3[p]\neq \str x = s_1[p-1]$, we know that $(p-1,p)\not\in\perm{s_1,s_3}$. This implies from Observation~\ref{obs:invariant} that $|s_3[1 \dd p-1]|_{\str{x}}=|s_1[1 \dd p-1]|_{\str{x}}=\ell$. But this means that $|s_3|_{\str{x}}=|s_3[1 \dd p-1]|_{\str{x}}+|s_3[p \dd n]|_{\str{x}}=\ell+r+1$ which contradicts $|s_3|_{\str{x}}=|s_1|_{\str{x}}=\ell+r$.

Finally, since $|s_1[1 \dd p-1]|_\str{x}=\ell$, $|s_1[p+1 \dd n]|_\str{x}=r$, and $|s_3[1 \dd p-1]|_\str{x}=\ell-1$, $|s_3[p+1 \dd n]|_\str{x}=r+1$, for any string $t$ matching $s_1$ and $s_3$ we necessarily have $t[p]=\str{x}$ (so that the last occurrence of $\str{x}$ on the prefix of $s_1$ can be swapped to the suffix of $s_3$). We deduce that $(p-1,p)\in\perm{s_1,t}$ and $(p,p+1)\in\perm{t,s_3}$, hence $t[p-1 \dd p+1]=s_1[p] \str x s_3[p]$. Note that this uses the assumption  $(p-1,p)\in \perm{s_1,s_2}$: in the other case ( $(p-1,p)\in \perm{s_2,s_3}$), we have $t[p-1 \dd p+1]=s_3[p] \str x s_1[p]$.
\qed \end{proof}

\begin{figure}
    \centering
    \includegraphics[width=0.2\linewidth]{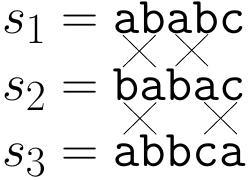}
    \caption{Example illustrating Lemma~\ref{lem:three-way-match}. Black crosses represent swaps between pairs of adjacent letters. One has $h_{s_1,s_2}=\str{1010}$, $h_{s_2,s_3}=\str{1001}$ and $h_{s_1,s_2}\oplus h_{s_2,s_3}=\str{0011}$. By Lemma~\ref{lem:three-way-match}, any string $t$ that matches both $s_1$ and $s_3$ must have $\str{bac}$ as a suffix---this is because it is the only suffix that can be swapped both into $\str{abc}$ and $\str{bca}$. For example, $t=\str{abbac}$ matches both $s_1$ and $s_3$. In this case, one has $h_{t,s_1}=\str{0010}$, $h_{t,s_3}=\str{0001}$, and $h_{t,s_2}=h_{t,s_1}\oplus h_{s_1,s_2}=h_{t,s_3}\oplus h_{s_3,s_2}=\str{1000}$.%
    }
    \label{fig:three-way-match}
\end{figure}

\paragraph*{Distances.}
Given two strings $s_1$ and $s_2$, the \emph{Adjacent Swap Distance} $\dS(s_1,s_2)$ (or simply \emph{Swap Distance} in this paper) is $+\infty$ if $s_1$ and $s_2$ are not matching, and the size of $\perm{s_1,s_2}$ otherwise. The \emph{Swap+Hamming distance} $\dSH(s_1,s_2)$ is the minimum over all strings $t$ of $\dS(s_1,t)+\dH(t,s_2)$, where the Hamming distance $\dH$ is the number of positions $p$, $1\leq p\leq n$, such that $s'[p]\neq t[p]$. 
Note that $\dSH$ is a lower bound on both the swap and the Hamming distances.

\section{Algorithms for the Swap Distance only}\label{sec:ds}
\begin{figure}
    \centering
    
    \begin{tikzpicture}
    \matrix [matrix of nodes,
             every node/.style={font=\tt},
            row sep=0] (m)
{
{}&\\[.1em]
&$s_1=$ &a&b&g&a&b&c&a&h&i&d&a&b&d&e&f&e&d&a&$\dS(s_1,s^*)\leq 4$\\
&$s_2=$ &b&a&g&c&a&a&b&i&h&d&a&b&e&f&d&d&e&a&$\dS(s_2,s^*)\leq 4$\\
&$s_3=$ &b&a&g&c&a&b&a&i&h&d&b&a&e&f&d&e&d&a&$\dS(s_3,s^*)\leq 4$\\[.5em]
{}&\\[.4em]
{}&{}&{}&{}&{}&{}&{}&{}&{}&{}&{}&{}&{}&{}&{}&{}&{}&{}&{}&{}& \ \\[.5em]
{}&\\[.1em]
&$s^*=$ &?&?&?&a&c&b&a&?&?&?&?&?&e&d&f&?&?&?& \\[.5em]
{}&{}\\[.1em]
&$s_1'=$&a&b&g&a&c&b&a&h&i&d&a&b&e&d&f&e&d&a&$\dS(s_1',s^*)\leq 2$\\
&$s_2'=$&b&a&g&a&c&b&a&i&h&d&a&b&e&d&f&d&e&a&$\dS(s_2',s^*)\leq 1$\\
&$s_3'=$&b&a&g&a&c&b&a&i&h&d&b&a&e&d&f&e&d&a&$\dS(s_3',s^*)\leq 2$\\[.3em]
{}&{}\\[.1em]
&$h_1 =$&0&0&0&0&0&0&0&0&0&0&0&0&0&0&0&0&0&.&$\dH(h_1,h^*)\leq 2$\\
&$h_2 =$&1&0&0&0&0&0&0&1&0&0&0&0&0&0&0&1&0&.&$\dH(h_2,h^*)\leq 1$\\
&$h_3 =$&1&0&0&0&0&0&0&1&0&0&1&0&0&0&0&0&0&.&$\dH(h_3,h^*)\leq 2$\\[.5em]
{}&{}\\[.1em]
&$h^* =$&1&0&0&0&0&0&0&1&0&0&0&0&0&0&0&0&0&.\\[.5em]
{}&{}\\[.1em]
&${s^*= }$&b&a&g&a&c&b&a&i&h&d&a&b&e&d&f&e&d&a&\\[.5em]
{}&{}\\[.1em]
&$s_1=$ &a&b&g&a&b&c&a&h&i&d&a&b&d&e&f&e&d&a&$\dS(s_1,s^*)= 4$\\
&$s_2=$ &b&a&g&c&a&a&b&i&h&d&a&b&e&f&d&d&e&a&$\dS(s_2,s^*)= 4$\\
&$s_3=$ &b&a&g&c&a&b&a&i&h&d&b&a&e&f&d&e&d&a&$\dS(s_3,s^*)= 3$\\
};
        
    \node[anchor=west] at (m-1-1.west) {1. Input strings};
    \node[anchor=west] at (m-5-1.west) {2. Tangled intervals (input strings are pairwise compatible outside these)};
    \draw[red, |-|] (m-6-6.west) -- (m-6-9.east);
    \draw[red,|-|] (m-6-15.west) -- (m-6-17.east);
    \node[anchor=west] at (m-7-1.west) { 3.   Any string in $C_\Scal$ is necessarily of the form:};
    \draw[red] (m-8-6.north west) rectangle (m-8-9.south east);
    \draw[red] (m-8-15.north west) rectangle (m-8-17.south east);

    \node[anchor=west] at (m-9-1.west) {4. Disentanglement};
    \fill[white, opacity=.6] (m-10-6.north west) rectangle (m-12-9.south east);    
    \fill[white, opacity=.6] (m-10-15.north west) rectangle (m-12-17.south east);

    \node[anchor=west] at (m-13-1.west) {5. With pairwise compatible strings, compute $h_i=h_{s_1', s_i'}$};
    \node[anchor=west] at (m-17-1.west) {6. Compute an $\anyConsensus{r,\dH}$ of $\{h_1,h_2,h_3\}$};
    
    \node[anchor=west] at (m-19-1.west) {7. Retrieve $s^*$ such that $h_*=h_{s_1', s^*}$};
    
    \node[anchor=west] at (m-21-1.west) {8. Check swap distance between $s^*$ and input strings};

    \foreach \pairs [count=\i from 22] in {%
        {3/4,7/8,10/11,15/16},
        {6/7,8/9,16/17,18/19},
        {6/7,12/13,16/17}%
        } {
    \foreach \u/\v in \pairs {
        \draw[red] (m-\i-\u.north west) rectangle (m-\i-\v.south east);
    }
    }
    
    \end{tikzpicture}
    \caption{Illustration of the algorithm for $\anyConsensus{r,\dS}$ over length-18 input strings  $\{s_1,s_2,s_3\}$, with $d=4$.
    } 
    \label{fig:dS-algorithm}
\end{figure}
We present our algorithm for distance $\dS$, i.e., for the swap distance only. See Figure~\ref{fig:dS-algorithm} for an illustration.

We first consider the case where all strings are pairwise matching (corresponding to steps 5-7 in the figure). In this case, we show that each of the studied consensus problems under the swap distance reduces to the same problem under the Hamming distance.%

Given a set of strings $\Scal=\{s_1, \dd , s_k\}$ all of length $n$, we call $\mathcal{C}_\Scal$ the set of strings from $\Sigma^n$ that are matching each and all strings in $\Scal$. We use the following result:

\begin{restatable}{lemma}{hamtosw}\label{lem:ham-to-sw}
Given a set of pairwise matching strings $\Scal=\{s_1, \dd , s_k\}$, for any string  $s^*\in \mathcal{C}_\Scal$, and any $s_i\in\Scal$, we have $\dS(s^*,s_i)=\dH(\swapstr{s_1, s^*}, \swapstr{s_1, s_i})$.
\end{restatable}
\begin{proof}
One can observe that for any pair of binary strings $h_1,h_2$, one has $\dH(h_1,h_2)=|h_1\oplus h_2|_1$, since $(h_1\oplus h_2[p])=\str{1}$ if and only if $h_1[p]\neq h_2[p]$. One also observes that, since the swap strings have ones exactly on the swap positions, we have $\dS(s_1,s_2)=|\swapstr{s_1,s_2}|_{\str{1}}$. Finally, since swaps are involutions, we also have $h_{u,v}=h_{v,u}$ for any pair of matching strings $u,v$. 

The result is now immediate: if $s_1,s^*,s_i$ are pairwise matching, Lemma~\ref{lem:three-way-match}, we obtain the following: $\dS(s^*,s_i)=|\swapstr{s^*,s_i}|_1=|\swapstr{s^*,s_1}\oplus \swapstr{s_1,s_i}|_1=\dH(\swapstr{s_1, s^*}, \swapstr{s_1, s_i})$.
\qed \end{proof}

\begin{corollary}\label{cor:pairwise-reduction}
Given a set of pairwise matching strings $\mathcal{S}=\{s_1, \dd, s_k\}$ all having length $n$, one can compute in linear time a set $\mathcal{S}_h=\{\swapstr{s_1, s_1} \dd \swapstr{s_1,s_k}\}$ of binary strings, all having length $n$, such that $\mathcal{S}$ admits a \rSconsensus (resp.~$\sSconsensus$, resp.~$\rsSconsensus$) for $d$ (resp.~for $D$, resp.~for $(d,D)$) if and only if $\mathcal{S}_h$ admits a $\anyConsensus{r,\dH}$ (resp.~$\anyConsensus{s,\dH}$, resp.~$\anyConsensus{rs,\dH}$) for $d$ (resp.~for $D$, resp. for $(d,D)$).
\end{corollary}

The latter result gives us a parameterized reduction to Hamming distance for pairwise matching strings. 

We consider now the more general case where input strings are not pairwise matching. Note that this does not trivially lead to no-instances, since only the solution needs to match with all input strings (not necessarily other input strings). However, a necessary condition for a yes-instance is that $\mathcal{C}_\Scal\neq\emptyset$.

We show a two‐step method for solving swap‐distance consensus problems. First, in linear time we transform $\Scal = \{s_1, \dots, s_k\}$ into $\Scal' = \{s'_1, \dots, s'_k\}$ by a minimal set of swaps so that the $s'_i$ are pairwise matching. Then, by Corollary~\ref{cor:pairwise-reduction}, we reduce $\Scal'$ to a variant of the consensus problem under Hamming distance, and prove that this variant admits a polynomial reduction to the standard Hamming‐distance consensus problem.

Let $\Scal=\{s_1, \dd , s_k\}$. If for some $s\in \mathcal{S}$ one has $(p,p+1)\in\perm{s^*,s}$ for every $s^*\in\mathcal{C}_\Scal$, we say that the swap $(p,p+1)$ is \emph{necessary} for $s$. We say that a set $\Scal'=\{s'_1, \dd, s'_k\}$ is a \emph{disentanglement} of $\Scal$ if \textbf{(i)} the strings from $\Scal'$ are pairwise matching, and \textbf{(ii)} for each $i$, the strings $s_i$ and $s_i'$ are matching and the permutation $\perm{s_i,s_i'}$ contains only necessary swaps for $s_i$. %

\begin{lemma}\label{lem:disentanglement}
Let $\Scal=\{s_1, \dd , s_k\}$ be a set of strings such that $\mathcal{C}_\Scal\neq \emptyset$. Then, $\Scal$ admits a disentanglement, which can be constructed in $\cO(kn)$ time.
\end{lemma}

\begin{proof}

Assume $\mathcal{C}_\Scal \neq \emptyset$, and let $s^* \in \mathcal{C}_\Scal$. Although we do not have access to $s^*$, its existence provides structural information about $\Scal$.

An interval $I \subseteq [1 \dd n]$ of length at least $2$ is \emph{tangled} if for every $i \in I$, there exists $s_j \in \Scal$ such that $h_{s_j, s^*}[i] = \str{1}$. If no tangled interval exists, then by Lemma~\ref{lem:three-way-match} all strings in $\Scal$ are pairwise matching and there are no necessary swaps.

If $I$ is tangled, Lemma~\ref{lem:three-way-match} implies that all strings in $\mathcal{C}_\Scal$ agree on $I$. Therefore, any difference between $s_j \in \Scal$ and $s^*$ within $I$ corresponds to a necessary swap. Importantly, this implies that the set of tangled intervals is independent of the choice of $s^* \in \mathcal{C}_\Scal$.

Assume a tangled interval exists, and let $i_0$ be the initial position of the first one. By definition, there exists $s_j \in \Scal$ with $h_{s_j, s^*}[i_0] = \str{1}$, hence $s^*[i_0] \neq s^*[i_0+1]$. Similarly, $s^*[i_0+1] \neq s^*[i_0+2]$. Let $s^*[i_0,i_0+1] = \str{ab}$. Furthermore, no $s_j\in\Scal$ admits a swap at position $i_0-1$ with $s^*$. We consider two cases:
\begin{itemize}
    \item If $s^*[i_0+2] = \str{c} \neq \str{a}$, then $\Scal[i_0,i_0+1] = \{\str{ba}, \str{ac}\}$.
    \item If $s^*[i_0+2] = \str{a}$, then $\Scal[i_0,i_0+1] = \{\str{ba}, \str{aa}\}$.
\end{itemize}
In either case, a necessary swap at position $i_0+1$ occurs for $s_j \in \Scal$ if and only if $s_j[i_0+1] \notin \Scal[i_0]$ or $s_j[i_0] = s_j[i_0+1]$. In particular, for such $s_j$, one has $s'[i_0+1,i_0+2]=s^*[i_0+1,i_0+2]=s_j[i_0+2,i_0+1]$ for every $s'\in \mathcal{C}_\Scal$.

We now describe how to compute $\Scal'$ (see Figure~\ref{ex:disentanglement} for an example). First, we can detect whether a tangled interval exists by checking if the set of all dirty columns can be grouped in pairs $i,i+1$ such that $\Scal[i,i+1] = \{\str{ab}, \str{ba}\}$ for distinct letters $\str{a}, \str{b}$. If this holds for all dirty columns, then there is no tangled interval and we set $\Scal' = \Scal$.

Otherwise, let $i_0$ be the first column violating this condition. Namely, the column $i_0$ is dirty and there is at least one $s_j\in\Scal$ with $s_j[i_0+1]\not\in \Scal[i_0]$ or $s_j[i_0+1]=s_j[i_0]$. 
Then, $i_0$ is the initial position of a tangled interval. Based on the earlier characterization, we can deduce $s^*[i_0+1,i_0+2]$ for any $s^*\in\mathcal{C}$, hence all necessary swaps at $i_0$: namely, we swap $s_j$ at $i_0$ if and only if $s_j[i_0,i_0+1]=s^*[i_0+1,i_0+2]$. After applying these swaps, we identify the strings requiring swaps at $i_0$—namely, those not matching at $i_0+1$ with the updated strings. At this step, we computed $s^*[i_0,i_0+1,i_0+2]$ and every necessary swap at position $i_0$ and $i_0+1$. If, after those swaps, there are still strings that do not match with $s^*$ at position $i_0+2$, we know that these have a necessary swap at position $i_0+2$, and hence we also know $s^*[i_0+3]$. By continuing this procedure as long as unmatched strings remain—progressively resolving each next position—we eventually reach the end of the tangled interval, or the end of $\Scal$.  Let $\Tilde{\Scal}$ be the set obtained after performing each of those necessary swaps. Then $\Tilde{\Scal}$ has no tangled interval before position $i_0+2$.

This process takes $\cO(k)$ time, and by repeating it at most $\cO(n)$ time, we eventually construct $\Scal'$ with no tangled intervals.

\qed \end{proof}

\begin{figure}\label{ex:disentanglement}
    \centering
    
    \begin{tikzpicture}
    \matrix [matrix of nodes,
             every node/.style={font=\tt, 
        text height=1.5ex,
        text depth=.25ex,},
            row sep=0
            ] (m)
{
{}&\\[.1em]
&$s_1=$ &g&a&b&c&a&h&i\\
&$s_2=$ &g&c&a&a&b&i&h\\
&$s_3=$ &g&c&a&b&a&i&h\\[.1em]
{}&\\[.1em]
&$s^*=$ &?\\
&$s_1=$ &g&a&b&c&a&h&i\\
&$s_2=$ &g&c&a&a&b&i&h\\
&$s_3=$ &g&c&a&b&a&i&h\\[.1em]
{}&\\[.1em]
&$s^*=$ &?&?&c&b\\
&$s_1'=$ &g&a&c&b&a&h&i\\
&$s_2'=$ &g&c&a&a&b&i&h\\
&$s_3'=$ &g&c&a&b&a&i&h\\[.1em]
{}&\\[.1em]
&$s^*=$ &?&a&c&b\\
&$s_1''=$ &g&a&c&b&a&h&i\\
&$s_2''=$ &g&a&c&a&b&i&h\\
&$s_3''=$ &g&a&c&b&a&i&h\\[.1em]
{}&\\[.1em]
&$s^*=$ &?&a&c&b&a\\
&$s_1'''=$ &g&a&c&b&a&h&i\\
&$s_2'''=$ &g&a&c&b&a&i&h\\
&$s_3'''=$ &g&a&c&b&a&i&h\\[.1em]
{}&\\[.1em]
{}&\\[.1em]
&$s^*=$ &?&a&c&b&a&?&?\\
};
        
    \node[anchor=west] at (m-1-1.west) {Input strings};
    \draw[blue] (m-2-3.north west) rectangle (m-4-3.south east);
    \node[anchor=west] at (m-5-1.west) {Column $i_0=1$ is clean, so not part of a tangled interval};
    
    \draw[blue] (m-7-4.north west) rectangle (m-9-4.south east);
    \draw[red] (m-7-5.north west) rectangle (m-7-5.south east);
    \node[anchor=west] at (m-10-1.west) {Column $i_0=2$ is dirty with $s_1[3]\not\in\Scal[2]=\{\str{a},\str{c}\}$, $\rightarrow$ $s^*[3,4]=s_1[4,3]=\str{cb}$};

    \draw[red] (m-12-5.south west) rectangle (m-14-5.south east);
    \draw[blue] (m-11-5.north west) rectangle (m-12-5.south east);
    \node[anchor=west] at (m-15-1.west) {Column $i_0=3$: $s_2'[3]=s_3'[3]=\str{a}\neq s^*[3]$ $\rightarrow s^*[2]=\str{a}$.};
    
    \draw[blue] (m-16-6.north west) rectangle (m-19-6.south east);
    \draw[red] (m-18-6.north west) rectangle (m-18-6.south east);
    \node[anchor=west] at (m-20-1.west) {Column $i_0=4$: $s_2''[4]=\str{a}\neq s^*[4]$ $\rightarrow$ $s^*[5]=\str{a}$.};

    \draw[blue] (m-21-7.north west) rectangle (m-24-7.south east);
    \node[anchor=west] at (m-25-1.west) {Column $i_0=5$: end of tangled interval};
    \node[anchor=west] at (m-26-1.west) {Columns $i_0=6$, $i_0=7$: no new tangled interval};
    \end{tikzpicture}
    \caption{Step-by-step disentanglement algorithm. We start with $\Scal=\{s_1,s_2,s_3\}$. Column $2$ is the first dirty column and we have $s_1[3]\not\in\Scal[2]=\{\str{a},\str{c}\}$. This means that there is a necessary swap at position $3$ for $s_1$, and in particular $s^*[3,4]=s_1[4,3]$ for \emph{every} $s^*\in\mathcal{C}_\Scal$. After swapping $s_1$ at position $3$, we obtain a new set $s'_1,s'_2,s'_3$, and we deduce the necessary swaps at position $2$, as the ones that do not match $s^*[2]$: namely, $s_2$ and $s_3$, and we deduce $s^*[2]=\str{a}$. After swapping $s'_2$ and $s'_3$ we obtain $\Scal''=s''_1,s''_2,s''_3$. We now resolve necessary swaps at position $4$, since $s^*[4]=\str{b}$ is known and $s''_2[4]=\str{a}\neq \str{b}$. We finally obtain the last set $\Scal'''=\{s'''_1,s'''_2,s'''_3\}$, and we see that every string matches $s^*$ at the new computed position $s^*[5]=\str{a}$. This means that we reached the end of the first tangled interval. Finally, we notice that the last set does not contain any other tangled interval, so it is a disentanglement for $\Scal$.
    } 
    \label{fig:disentanglement-algorithm}
\end{figure}%

The following result is deduced from Lemma~\ref{lem:disentanglement}:

\begin{corollary}\label{cor:s-swap}
\sSconsensus can be solved in $\cO(kn)$ time.
\end{corollary}
\begin{proof}
From Lemma~\ref{lem:disentanglement}, we can perform every necessary swap in $\cO(kn)$. Let $\Delta$ be the number of swaps performed in total at this step. The set obtained contains only pairwise matching strings, hence, from Corollary~\ref{cor:pairwise-reduction} we can reduce the problem to $\anyConsensus{s,\dH}$ in $\cO(kn)$: the set $\Scal$ admits a \sSconsensus for $D$ if and only if its disentanglement admits a $\anyConsensus{s,\dH}$ for at most $D-\Delta$. The result then follows since \anyConsensus{s,\dH} can be solved in linear time~\cite{GRA}.
\qed \end{proof}

For $\rSconsensus$ and $\rsSconsensus$, one needs an additional result in order to use Corollary~\ref{cor:pairwise-reduction} and reduce the problems to their Hamming versions. Indeed, the strings in the disentanglement have been swapped a different number of times, and we need to consider this when computing the consensus. This can be formalized as follows:

\probdef{\rHConsensusM}{A set of $k$ strings $\Scal=\{s_1, \dd , s_k\}$ each of length $n$, integer~$d$ and for each $s\in \Scal$, an integer $x_s\le d$.}{Is there a length-$n$ string~$s^*$ such that $\forall~s\in \Scal,~\displaystyle  \dH(s,s^*)\leq d-x_s~$?}

\probdef{\rsHConsensusM}{A set of $k$ strings $\Scal=\{s_1, \dd , s_k\}$ each of length $n$, integers~$d,D$ and for each $s\in \Scal$, an integer $x_s\le d$.}{Is there a length-$n$ string~$s^*$ such that $\forall~s\in \Scal,~\displaystyle  \dH(s,s^*)\leq d-x_s~$ and $\displaystyle \sum_{s\in \Scal} \partial(s,s^*)\leq D-\sum_{s\in\Scal}x_s$?}

The following is a direct consequence of Lemma~\ref{lem:disentanglement} and Corollary~\ref{cor:pairwise-reduction} (where \emph{parameterized reduction} refer to the parameters $d$, $k$, $n$ and, when applicable, $D$):

\begin{corollary}
The problem $\rSconsensus$ (resp. $\rsSconsensus$) admits a parameterized reduction 
to $\rHConsensusM$ (resp.~to $\rsHConsensusM$). %
\end{corollary}

Observe that $\rHConsensusM$ and $\rsHConsensusM$ are clearly in $\NP$, hence they can be reduced to $\anyConsensus{r,\dH}$ (resp.~$\anyConsensus{rs,\dH}$), that are $\NP$-complete, in polynomial time. In fact, we can even show the following:

\begin{restatable}{lemma}{mixed-to-hamming}\label{lem:mixed to hamming}
The problem $\rHConsensusM$ (resp.  $\rsHConsensusM$) admits a parameterized reduction to $\anyConsensus{r,\dH}$ (resp.~$\anyConsensus{rs,\dH}$).
\end{restatable}
\begin{proof}
We prove the result for $\rsHConsensusM$, as the proof for $\rHConsensusM$ is identical, without the considerations on $D$.

Given a set of $k$ strings $\Scal=\{s_1, \dd , s_k\}$ each of length $n$, integers~$x_1 \dd x_k$ and $D$, we construct a corresponding instance of $\anyConsensus{rs,\dH}$. Let $x=\max_s \{x_s\}$. For each $s\in \Scal$ we define strings $a_s$ and $b_s$ over the alphabet $\Sigma\cup \{\str{0},\str{1}\}$ (where $\str{0},\str{1}\not\in\Sigma$) as follows: $a_s=\str{01}^{x_s} \cdot \str{00}^{x-x_s}$ and $b_s=\str{10}^{x_s} \cdot \str{00}^{x-x_s}$. We claim that $\Tilde{s}$ is a solution for \anyConsensus{rs,\dH} with $\Scal'=\{s\cdot a_s|~s\in\Scal\}\cup,\{s\cdot b_s|~s\in\Scal\}$ and parameters $(d,2D)$ if and only if $s^*=\Tilde{s}[1 \dd n]$ is a solution for $\rsHConsensusM$ on $\Scal$ with parameters $(d,D)$ and $x_1 \dd x_k$.

Let $s^*$ be a solution for \rsHConsensusM~on $\Scal$ with parameters $(d,D)$ and $x_1 \dd x_k$. Then one can extend $s^*$ into $\Tilde{s}=s^*\cdot \str{00}^x$ to obtain a solution for \anyConsensus{rs,\dH} with $\Scal'$ and parameters $(d,2D)$: indeed, for every $s\in\Scal$ we have $\dH(\str{00}^x,a_s)=\dH(\str{00}^x,b_s)=x_s$. 

Conversely, let $\Tilde{s}$ be a solution for \anyConsensus{rs,\dH} with $\Scal'$ and parameters $(d,2D)$. We write $s^*=\Tilde{s}[1 \dd n]$ and $s'=\Tilde{s}[n+1 \dd n+x]$. Since for each $s$ one has $\dH(a_s,b_s)=2x_s$, one has $\dH(s',a_s)+\dH(s',b_s)\ge 2x_s$ from the triangular inequality, without loss of generality $\dH(s',a_s)\ge x_s$. But then, one has $\dH(s^*,s)=\dH(\Tilde{s},s\cdot a_s)-\dH(s',a_s)\le d-x_s$.

Similarly, one can observe that $\sum_s \dH(s',a_s)+\sum_s \dH(s',b_s)\ge 2\sum_s{x_s}$ and hence $\sum_s \dH(s,s^*)\le D$.

Finally, we can conclude by observing that $|\Scal'|=2k(n+d)$. One can assume that $d\le n$ (if not, one can replace $d$ with $n$, each $x_s$ with $x_s-(d-n)$, and update $D$ accordingly), hence $|\Scal'|=\cO(kn)$.
\qed \end{proof}

We can now state the main theorem of the section:

\begin{theorem}\label{thm:results swap}
The problems $\rSconsensus$ (resp.~$\rsSconsensus$) and $\anyConsensus{r,\dH}$ (resp.  $\anyConsensus{rs,\dH}$) both admit a parameterized reduction to each other. We have the following: \begin{itemize}
    \item $\rSconsensus(d)\in \FPT$ and $\rsSconsensus(d)\in \FPT$.
    \item $\rSconsensus(k,d)$ and $\rsSconsensus(k,d)$ both have polynomial kernels of size $\cO(k^2 d \log k)$.
    \item $\rsSconsensus(D)$ has a polynomial kernel of size $\cO(D^3\log D)$.
    \item Neither $\rSconsensus(d,n,|\Sigma|)$ nor $\rsSconsensus(d,n,|\Sigma|)$ have polynomial kernels unless $\NP\subseteq \text{co}\NP/\text{poly}$.
\end{itemize}
\end{theorem}
\begin{proof}
Lemmas~\ref{lem:disentanglement}%
gives us a parameterized reduction from $\rSconsensus$ (resp.~$\rsSconsensus$) to $\anyConsensus{r,\dH}$ (resp.~to $\anyConsensus{rs,\dH}$), and~\cite{AMI} gives us the other direction for $\rSconsensus$ by providing a linear reduction between both problems. In fact, the construction used in~\cite{AMI} also implies a reduction from $\anyConsensus{rs,\dH}$ to $\rsSconsensus$, as all the distances are preserved. Since $\anyConsensus{r,\dH}(d)\in\FPT$~\cite{GRA}, and $\anyConsensus{rs,\dH}(d)\in\FPT$~\cite{BUL}, both problems are $\FPT$ for $d$. The claims on polynomial kernels follow from the one given for corresponding problems on Hamming distance~\cite{BasavarajuEtAl2018,BUL}.\qed  %

\end{proof}

\section{Algorithms for the Swap+Hamming distance}
\label{sec:dsh}

\subsection{Fixed parameter complexity of \rSHconsensus}

In this section, we study the parameterized complexity of $\rSHconsensus$. We start by stating two basic results:

\begin{restatable}{lemma}{greedydsh}\label{lem:greedy-dsh}
Given two strings $s_1,s_2$ both of length $n$, the distance $\dSH$ can be computed in linear time as follows: If $s_1[i,i+1]=s_2[i+1,i]$, we count a swap at position $i$ if and only if we did not count a swap at position $i-1$. The distance $\dSH$ is then the number of mismatches not resolved by a swap.
\end{restatable}
\begin{proof}
This greedy strategy gives a valid transformation. To prove optimality, observe that swapping at $i$ when $s_1[i,i+1] \ne s_2[i+1,i]$ leaves at least one mismatch, so only valid swaps reduce cost. Now, skipping a swap at $i$ to allow a later swap (e.g., in a pattern like $\str{ab}\dd\str{ba}$) gains at most one, matching the loss from skipping the swap at $i$. Thus, the greedy left-to-right strategy is optimal.

\qed \end{proof}

\begin{lemma}\label{lem:triangle inequality}
Given two strings $s_1,s_2$ both of length $n$, the following inequality holds:
$$\dSH(s_1,s_2)\le \dH(s_1,s_2)\le 2\dSH(s_1,s_2).$$

As a consequence, given $3$ strings $s_1,s_2,s_3$ all of length $n$, a \emph{weakened triangle inequality} holds:
$$\dSH(s_1,s_3)\le\min\{2\dSH(s_1,s_2)+\dSH(s_2,s_3),\dSH(s_1,s_2)+2\dSH(s_2,s_3)\}.$$
\end{lemma}

\begin{proof}
The first statement follows from the fact that a swap transformation modifies two letters, and that any transformation that performs only substitutions is also valid for the Swap+Hamming distance (even if it is not optimal in general).

For the second statement, one may notice that given $s_1,s_2,s_3$, it holds that $\dSH(s_1,s_3)\le \dSH(s_1,s_2)+\dH(s_2,s_3)$. Indeed, composing a transformation that performs swaps and substitutions and a transformation that performs only substitutions gives a valid Swap+Hamming transformation. However, this does not hold if both transformations involve swaps, as those swaps might be incompatible. The second statement then follows from the first one by noticing that the roles of $s_1$ and $s_3$ are symmetrical.
\qed \end{proof}

Our result is mainly based on the following Lemma, which is already central in the proof of the corresponding theorem for the Hamming distance.

\begin{lemma}[\cite{GRA}]\label{lem:getting closer}
Let $\Scal=\{s_1, \dd , s_k\}$ be a set of strings having length $n$, $s^*$ a string satisfying $\max_{s\in\Scal}\dH(s^*,s)\le d$, and $\Tilde{s}$ a string such that for at least one $s\in\Scal$ we have $\dH(\Tilde{s},s)\ge d+1$. Let us fix a set $P$ of $d+1$ positions $p$ such that $\Tilde{s}[p]\neq s[p]$ for each $p\in P$. Then, for at least one $p\in P$, $s[p]=s^*[p]$.
\end{lemma}

We now prove our main result on $\rSHconsensus$:

\begin{theorem}\label{thm:rSHconsensus FPT}
$\rSHconsensus(d)\in\FPT$.
\end{theorem}

\begin{proof}
    We adapt the algorithm from \cite{GRA}, with additional branchings that take swapped pairs of letters from one of the strings in $\Scal$. Namely, let us assume that we are given a candidate string $\Tilde{s}$ such that, for some fixed $s\in\Scal$, we have $\dSH(\Tilde{s},s)>d$, and that there exists a consensus string $s^*$ with $\max_{s\in\Scal}\dSH(s,s^*)\le d$. We first show that, similarly to the Hamming distance, we can always find a set of $\cO(d)$ strings where at least one is strictly closer (in Hamming distance) from $s^*$ than $\Tilde{s}$ is.
    
    Let us first assume that $\dH(\Tilde{s},s)\ge 2d+1$. Then, from Lemma~\ref{lem:triangle inequality}, we know that $\max_{s\in\Scal}\dH(s,s^*)\le 2d$, hence from Lemma~\ref{lem:getting closer}, for any fixed set $M$ of $2d+1$ mismatch positions between $\Tilde{s}$ and $s$, the strings $s$ and $s^*$ must match on at least one position from $M$. We can then obtain a suitable set by defining for each $p\in M$ a string obtained from $\Tilde{s}$ by replacing $\Tilde{s}[p]$ by $s[p]$. Let us now assume that $\dH(\Tilde{s},s)\le 2d$. By hypothesis and by Lemma~\ref{lem:triangle inequality}, we also have $\dH(\Tilde{s},s)\ge d+1$. Let $M$ be the set of mismatching positions between $\Tilde{s}$ and $s$, with $|M|=\cO(d)$. If for some $p\in M$, one has $s[p]=s^*[p]$, we obtain the set of $\cO(d)$ strings as before. On the other hand, if for every $p\in M$, one has $s[p]\neq s^*[p]$, then the strings $s$ and $s^*$ have $d+1$ mismatches, and since $\dSH(s,s^*)\le d$, there is at least one swap between $s$ and $s^*$ involving a position from $M$. So for some pair $(p,p+1)$ with $p$ or $p+1$ in $M$, one has $s[p,p+1]=s^*[p+1,p]$. What remains to show is that for at least one such pair, one has $\Tilde{s}[p,p+1]\neq s^*[p,p+1]$: if this is not the case, then each mismatch (resp.~swap) from $\Tilde{s}$ to $s$ is also a mismatch (resp.~swap) from $s$ and $s^*$, hence $\dSH(s,s^*)\geq\dSH(s,s')\ge d+1$, a contradiction. In that case, the $\cO(d)$ strings are constructed by replacing $\Tilde{s}[p,p+1]$ with $s[p+1,p]$, or $\Tilde{s}[p-1,p]$ with $s[p,p-1]$ for each $p\in M$.
    
       We can now describe an \FPT algorithm for \rSHconsensus: we start with a candidate $\Tilde{s}\in \Scal$ and explore each branch as described (if $\dH(\Tilde{s},s)\ge 2d+1$, we consider only $2d+1$ substitutions, otherwise we also consider the swaps as detailed above). There are $\cO(d)$ of them since we impose $|M|\le 2d+1$. The depth of the branching tree is also in $\cO(d)$: since $\Tilde{s}\in\Scal$, one has $\dH(\Tilde{s},s^*)\le 2 \dSH(\Tilde{s},s^*)\le 2d$ (Lemma~\ref{lem:triangle inequality}), and we can conclude because this distance decreases for at least one branch at each step. Therefore, after $i$ recursive calls, the best candidate string is at Hamming distance at most $2d-i$ from a solution, hence at Hamming distance at most $4d-i$ from any string in $\Scal$ (by Lemma~\ref{lem:triangle inequality} and the standard triangle inequality for the Hamming distance). Therefore, we can ignore the search branches that cannot lead to a solution by trimming any instance such that $\dH(\Tilde{s},s)\ge 4d-i+1$, where $i$ is the depth of the current recursion in the search tree.\qed

\end{proof}

Finally, we study the kernelization of the consensus problems under $\dSH$. Intuitively, problems for the \dSH~distance are at least as hard as the corresponding problems for the Hamming distance; hence, we can adapt the results from~\cite{BasavarajuEtAl2018,BUL}.

\begin{theorem}\label{thm:SH no polyk}
Neither $\rSHconsensus(d,n,|\Sigma|)$, nor $\rsSHconsensus(d,n,|\Sigma|)$ admits a polynomial kernel, unless $\NP\subseteq \text{co}\NP/\text{poly}$. 
\end{theorem}%
\begin{proof}
Given an instance $\Scal$ of $\anyConsensus{r,\dH}$ (resp.~$\anyConsensus{rs,\dH}$), one constructs an instance $\Scal'$ for $\rSHconsensus$ (resp.~$\rsSHconsensus$) by inserting between each column a bogus column that contains a single letter $\$$ that is not in the alphabet. It is clear that given a consensus $s^*$ (for any of the two problems), there is no swap between $s^*$ and any of the $s\in \Scal'$ since the sets of letters on even and odd positions in $\Scal'$ are disjoint. Hence, the Swap+Hamming distances from $s^*$ to the strings in $\Scal'$ exactly correspond to Hamming distances between the strings in $\Scal$ and the string obtained by removing the odd positions in $s^*$, namely, $\Scal'$ is a yes-instance for $\rSHconsensus$ (resp.~$\rsSHconsensus$) if and only if $\Scal$ is a yes-instance for $\anyConsensus{r,\dH}$ (resp.~$\anyConsensus{rs,\dH}$). Since the problems \anyConsensus{r,\dH}~\cite{BasavarajuEtAl2018} and \anyConsensus{rs,\dH}~\cite{BUL} have no polynomial kernels for $(d,n)$ unless $\NP\subseteq \text{co}\NP/\text{poly}$, we conclude with the claimed result.
\qed \end{proof}
\subsection{Polynomial algorithm for \sSHconsensus}

We conclude by showing that \sSconsensus can be solved in polynomial time, using a dynamic programming algorithm. %

Let $\Scal=\{s_1, \dd , s_k\}$ be a set of strings all having length $n$, and let $s^*$ be an arbitrary string of length $n$. We write $SW(\Scal,s^*,i)$ for the set of $j\in[1 \dd k]$ such that there is a swap at position $i$ between $s^*$ and $s_j$ (swaps are taken greedily, from left to right, as in Lemma~\ref{lem:greedy-dsh}). We say that $W=SW(\Scal,s^*,i)$ is \emph{reached} at position $i$ (or that $(i,W)$ is \emph{reached}) by $s^*$. Given a set $W\subseteq [1 \dd k]$, and a position $i\in [1 \dd n]$, we say that $(i,W)$ is \emph{reachable} from $\Scal$, or that $W$ is \emph{reachable} at position $i$ from $\Scal$ if it is reached by some $s^*$.

We start with the following Lemma. Informally, it implies that a solution for \sSHconsensus can be obtained from a solution for \anyConsensus{s,\dH} by the repeated operations of replacing its letters by swapped letters taken from strings in $\Scal$.

\begin{restatable}{lemma}{swapfromcenter}\label{lem:swap-from-center}
Let $\Scal=\{s_1, \dd , s_k\}$ be a set of strings all having length $n$ and admitting a solution to $\sSHconsensus$ with sum of distances $D$. Let $s^*_H$ be the lex-minimal solution to $\anyConsensus{s,\dH}$. Then, if $s^*_{SH}$ is the lex-minimal solution to $\sSHconsensus$ for $\Scal$, the following holds: for each $1\le i\le n$, if one has $s^*_{SH}[i]\neq s^*_H[i]$, then one has $SW(\Scal,s^*,i-1)\neq \emptyset$ or $SW(\Scal,s^*,i)\neq \emptyset$.%
\end{restatable}
\begin{proof}
Let $i$ be such that $s^*_{SH}[i]\neq s^*_H[i]$. If for every $j$, the strings $s^*_{SH}$ and $s_j$ have no swap at positions $i-1$, nor at position $i$, then one has:
\begin{align*}
    \sum_{j\le n} \dSH(s^*_{SH},s_j)&=\sum_{j\le n} \dSH(s^*_{SH}[1 \dd i-1],s_j[1 \dd i-1])\\&+\sum_{j\le n}\dH(s^*_{SH}[i],s_j[i])\\&+\sum_{j\le n}\dSH(s^*_{SH}[i+1 \dd n],s_j[i+1 \dd n])
\end{align*}

But by definition of $s^*_H$, one has $$\sum_{j\le n}\dH(s^*_{H}[i],s_j[i])\le \sum_{j\le n}\dH(s^*_{SH}[i],s_j[i]),$$ and by optimality of $s^*_H$ and $s^*_{SH}$, both sides are equal. Since each strings is also lex-minimal, we obtain  $s^*_{SH}[i]=s^*_{H}[i]$. 
\qed \end{proof}

Given $\Scal$, we consider the set $\Scal_i$ of length $i+1$ prefixes of strings in $\Scal$. We define a dynamic programming table $\Tf$ as follows: Let $i\in [n]$ and $W\subseteq[k]$ such that $(i,W)$ is reachable. Then, $\Tf[i,W]$ is, among all strings $t\in \Sigma^{i+1}$ reaching $(i,W)$, the one minimizing $\sum_{s\in \mathcal{S}_i} \dSH(s,t)$. In case of ties, we pick the lex-minimal such string. Additionally, we set $\Tf[0,\emptyset]=s_H^*[1]$, where $s_H^*$ is a lex-minimal solution for $\anyConsensus{s,\dH}(\Scal)$.

\begin{remark}\label{rk:unique-suffix}
Consider a set of $n$-length strings $\Scal=\{s_1,\dd,s_k\}$, and a pair $(i,W)$ reached by some $t\in \Sigma^{i+1}$. Furthermore, we assume $W \neq \emptyset$. By definition, one must have $t[i,i+1]=s_j[i+1,i]$ for every $j\in W$. In particular, when $W\neq \emptyset$, the pair $(i,W)$ uniquely determines the length-$2$ suffix of $\Tf[i,W]$. However, the converse does not hold, as illustrated in the following example:   
\end{remark}
\begin{example}\label{ex:same-suffix}
Let $\Scal=\{s_1,s_2\}$, with $s_1=bab$ and $s_2=aab$ and consider $t_1=aba$, $t_2=bba$. One can observe that in this case, $SW(\Scal,t_1,2)=\{2\}$, and $SW(\Scal,t_2,2)=\{1,2\}$, even if $t_1[2,3]=t_2[2,3]$.
\end{example}

Although there are $2^k$ distinct subsets $W\subseteq [k]$, the following lemma shows that at most $\cO(kn)$ of them are reachable at a given position $i$; hence the size of the table $\Tf$ remains polynomial.

\begin{lemma}\label{lem:poly-reachable}
Let $\Scal=\{s_1, \dd , s_k\}$ be a set of strings all having length $n$, and $1\le i\le n$. At each position $i$, there are at most $ki$ reachable sets. %
\end{lemma}
\begin{proof}
Write $N_i$ for the number of reachable sets at position $i$. For every $\str{ab}\in \Sigma^2$ we write $W_{\str{ab},i}=\{j\in[k]|~s_j[i,i+1]=\str{ba}\}$. There are at most $k$ such sets for each $i$, and the sets $W_{\str{ab},1}$ are exactly the reachable sets at position $1$, hence $N_1\le k$. 
Now, assume $N_i\le ki$. If $(i+1,W)$ is reached by a string $t$, with $t[i+1,i+2]=\str{ab}$, writing $W'=SW(\Scal,t,i)$, and $\overline{W'}=[k]\setminus W'$, by Lemma~\ref{lem:greedy-dsh} we have $W=W_{\str{ab},i+1}\cap \overline{W'}$. Hence, $W\neq W_{\str{ab},i+1}$ only if there exists $j\in W'\cap W$. But then we have $s_j[i+1,i+2]=\str{ba}$ (because $j\in W$), and $t[i]=s_j[i+1]=\str{b}$ (because $j\in W'$). This means that $t[i,i+1]=\str{ba}$. Since each reachable set $(i,W')$ swaps at most one pair of letters, there are at most $N_i$ sets that are reachable at position $i+1$ other than the sets $W_{\str{ab},i+1}$ for $\str{ab}\in \Scal[i+1,i+2]$. Hence, $N_{i+1}=N_i+k$. The claim then follows by induction. \qed
\end{proof}

The following lemma describes the dynamic programming relations for $\Tf$. We first define, for every reachable $(i,W)$, $i\ge 2$, the set $\pre(i,W)$ of strings that are used to compute $\Tf[i,W]$.

If $W\neq \emptyset$, from Remark~\ref{rk:unique-suffix}, the length-$2$ suffix of $\Tf[i,W]$ is entirely determined. On the other hand, if $W=\emptyset$, we can assume, from Lemma~\ref{lem:swap-from-center}, that the last letter of $\Tf[i,W]$ is $s_H^*[i+1]$, where $s_H^*$ is a lex-minimal solution for $\anyConsensus{s,\dH}(\Scal)$. In both cases, let us write $\str{b}$ for the last letter of $\Tf[i,W]$. We then set $t= \Tf[i-1,W']\cdot \str{b}\in\pre(i,W)$ if $t$ reaches $W$ at position $i$. Furthermore, if $W\neq \emptyset$, writing $\str{a}$ for the penultimate letter of $\Tf[i,W]$ (which is entirely determined from $W$), we also consider strings of the form $t=\Tf[i-2,W']\cdot \str{ab}$ where $(i-2,W')$ is reachable, and such that $t$ reaches $W$ at $i$ and $\emptyset$ at $i-1$, if such strings exist.

\begin{lemma}\label{lem:DP}
Let $\Scal=\{s_1,\dd s_k\}$ be a set of strings all having length $n$. For $i\ge 2$, the table $\Tf$ satisfies the following recurrence relation:
\[\Tf[i,W]=\argmin_{t\in \pre(i,W)}\sum_{s\in \Scal_{i}} \dSH(s,t).\]

Furthermore, the states $\Tf[1,W]$ can be computed for every $W$ reachable at position $1$ by setting $\Tf[0,\emptyset]=s_H^*[1]$, $\Tf[1,W]=\str{ba}$ where $W=\{j\in [k]|~s_j[1,2]=\str{ab}\}$ for every $\str{ab}\in \Scal[1,2]$, and $\Tf[1,\emptyset]=s_H^*[1,2]$, if it reaches $\emptyset$, where $s^*_H$ is chosen to be the lexicographically minimal solution for $\anyConsensus{s,\dH}$.
\end{lemma}

\begin{proof}
The initialisation of $\Tf$ is clear, because for those cases, the constraints on the admissible strings define a unique candidate.   

For the recurrence relation, we write $t=\Tf[i+1,W]$. We prove, by induction, that we necessarily have $t\in \pre(i,W)$, which directly implies the claim. Let us write $W'=SW(\Scal_i,t,i)$. First assume \(W'\neq \emptyset\). We write $\Tf[i,W]=t\cdot \str{b}$. Extending from $t$ to $t\cdot \str{b}$ adds a cost of $1$ for each $j$ such that $s_j[i+1] \neq \str{b}$, unless a swap at $i$ occurs (i.e., $j \in W$), which cancels the mismatch. The net change is the number of mismatches at $i+1$ minus the number of swaps at $i$; hence, we have
\[
\sum_{s\in S_i}\dSH(\Tf[i,W],s)
=
\sum_{s\in S_{i-1}}\dSH(t,s)
+|\{\,j\in[k]\colon s_j[i+1]\neq b\}|
-|W|.
\]
Then if $$
\sum_{s\in S_{i-1}}\dSH(t,s)
\ge
\sum_{s\in S_{i-1}}\dSH(\Tf[i-1,W'],s),$$ one has 
$$\sum_{s\in S_i}\dSH(\Tf[i,W],s)
\ge
\sum_{s\in S_i}\dSH(\Tf[i-1,W']\cdot b,s),$$ and both sides are equal by optimality of \(\Tf[i,W]\).  Since ties are broken lexicographically, it follows that $t=\Tf[i-1,W']$, and 
$
\Tf[i,W]
=t\cdot\str{b}
\in\pre(W)$.
Now if \(W'=\emptyset\), we write $	\Tf[i,W]=t\cdot ba$
(note that $\str{ab}$ is entirely determined by $W'$ and $W$).	
In that case, one obtains:
\[
\sum_{s\in S_i}\dSH(\Tf[i,W],s)
=\sum_{s\in S_{i-2}}\dSH(t,s)
+\sum_{s\in S}\dSH(\str{ab},s[i,i+1]).
\]
and by optimality of $\Tf[i,W]$, writing $\tilde{W}=SW(\Scal,t,i-2)$, we have $t=\Tf[i-2,\tilde{W}]$, and $T[i-2,\tilde{W}]\cdot\str{ab}$ reaches $\emptyset$ at $i-1$ and $W$ at $i$. Hence, $\Tf[i,W]=\Tf[i-2,\tilde{W}]\cdot\str{ab}\in\pre(i,W)$.
\end{proof}

\begin{theorem}\label{thm:sSHconsensus poly}
    One can solve \sSHconsensus in $\cO(k^3n^2)$ or $\cO(|\Sigma|k^2n^2)$ time.
\end{theorem}
\begin{proof}
We use dynamic programming with the table $\Tf[i,W]$. We first compute $s^*_H$, a lexicographically minimal solution for \anyConsensus{s,\dH}, and we initialize the lines $0$ and $1$ of $\Tf$ as described in Lemma~\ref{lem:DP}. We then compute $\Tf$ inductively using the relation from Lemma~\ref{lem:DP}. More precisely, for each reachable $(i,W)$ we construct all the strings that are extensions of $t=\Tf[i,W]$ and are contained in $\pre(i',W')$ for some $W'$ that is reachable at position $i'=i+1$ or $i'=i+2$. By definition, those strings are exactly the strings of the form $t\cdot \str{a}$ for $\str{a}\in \Scal[i+1]$ that reach $W\neq \emptyset$ at position $i+1$, the string $t\cdot s_H^*[i+2]$ if it reaches $\emptyset$ at position $i+1$, and the strings of the form $t\cdot \str{ab}$ reaching $\emptyset$ at position $i+1$, with $\str{ba}\in \Scal[i+2,i+3]$. Note that for the special case $\Tf[0,\emptyset]=s_H^*[1]$, we compute only this last type of extensions, as the other ones would correspond to elements of the first line of the table, which is already computed during initialization.

For each such $t'$, we compute its partial score $d_{t'}$ (the sum of distances up to $i+1$), and store the tuple $(i',W', t', d_{t'})$, where $W'$ is the set reached by $t'$ on position $i'=|t'|-1$. We maintain only the best-scoring string for each $(i',W')$, and for every position $i'$ we sort tuples lexicographically by the indicator vector of $W' \subseteq [k]$. By construction, and by Lemma~\ref{lem:DP}, after processing every $i<i'$, the only tuple corresponding to $(i',W')$ contains $\Tf[ i',W']$.

From Lemma~\ref{lem:poly-reachable}, we have at most $\cO(kn)$ partial solutions per step. Each of them can be extended in at most $\cO(k)$ or $\cO(|\Sigma|)$ ways (extensions are always taken from $\Scal[i,i+1]$, or $s_H[i]$). Updating distances and computing swap sets both take $\cO(k)$ time. Thus, the total time complexity is $\cO(k^3n^2)$, or $\cO(|\Sigma|k^2n^2)$. \qed 

\end{proof}

\begin{example}
Let $s_1=\str{baba},~s_2=\str{cabc},~s_3=\str{abca}$ and $\Scal=\{s_1,s_2,s_3\}$, with $s^*_H=\str{aaba}$. The initial states are $\Tf[0,\emptyset]=\str{a}$, $\Tf[1,\emptyset]=\str{aa}$, $\Tf[1,\{1\}]=\str{ab}$, $\Tf[1,\{2\}]=\str{ac}$, $\Tf[1,\{3\}]=\str{ba}$. When processing, for example, $\Tf[1,\{3\}]$, we first compute $t=\str{ba}\cdot s^*_H[3]=\str{bab}$ and add the tuple $(2,\emptyset,t,d_t)=(2,\emptyset,\str{bab},3)$ only if $t$ reaches $\emptyset$ at position $2$. This is the case here (in particular, $t$ does not reach $\{1\}$ because $t$ and $s_1$ already have a swap at position $1$). Then, we compute the strings $\str{ba}\cdot x$ for $x\in \Scal[2]$, namely $\str{baa}$ and $\str{bab}$, but both reach $W=\emptyset$ at position $2$ so we do nothing. Finally, we compute $\str{ba}\cdot y \cdot x$ where $x \cdot y\in \Scal[3,4]$, obtaining $\str{baab},~\str{bacb},~\str{baac}$. Among those, only $\str{bacb}$ reaches $\emptyset$ at position $2$, and it reaches $W=\{2\}$ at position $3$, hence we create the tuple $(3,\{2\},\str{bacb},6)$ if we do not already have a tuple $(3,\{2\},t',d_t)$ with $d_t\le 6$.

In this example, the final solution $\str{baba}=\Tf[3,\emptyset]$ is found by extending $\str{bab}=\Tf[2,\emptyset]$, and $\sum_{s\in \Scal}\dSH(s,\str{baba})=4$. See Table~\ref{tab:running_full} for the full DP-table explored by the algorithm. Note that $(2,\{1,2\})$ is also reachable, for example by $\str{caba}$, but this string is also reaching $\emptyset$ at position $1$, and since $\texttt{c}\neq s_H^*[1]$ it can be ignored by Lemma~\ref{lem:ham-to-sw}, and is indeed never computed in the algorithm.

\begin{table}[h]
	\centering
	\caption{DP‐states \(T[i,W]\) for 
		\(S=\{\str{baba},\str{cabc},\str{abca}\}\), 
		\(s_H^*=\str{aaba}\).}
	\label{tab:running_full}
	\begin{tabular}{l@{\quad}|@{\quad}l@{\quad}l@{\quad}l@{\quad}l}
		\(W\)\(\backslash\)\(i\) 
		& \(0\)        & \(1\)        & \(2\)                                    & \(3\)                                          \\ \hline
		\(\emptyset\) 
		& \(\str{a}\) & \(\str{aa}\) & \(\str{bab}=T[1,\{3\}]\cdot s_H^*[3]\)
		& \(\str{baba}=T[2,\emptyset]\cdot s_H^*[4]\)       \\
		\(\{1\}\)   
		&             & \(\str{ab}\) & 
		& \(\str{abab}=T[2,\{2\}]\cdot\str{b}\)           \\
		\(\{2\}\)   
		&             & \(\str{ac}\) & \(\str{aba}=T[1,\{1\}]\cdot\str{a}\)
		& \(\str{bacb}=T[1,\{3\}]\cdot\str{cb}\)          \\
		\(\{3\}\)   
		&             & \(\str{ba}\) & \(\str{abc}=T[1,\{3\}]\cdot\str{c}\)
		& \(\str{abac}=T[2,\{2\}]\cdot\str{c}\)           \\
	\end{tabular}
\end{table}
\end{example}
\section{Open questions}

We leave the following questions open: 
\begin{itemize}
    \item Is $\rsSHconsensus(d)\in \FPT$?
    \item Can one find conditional lower bounds, or a more efficient algorithm, for $\sSHconsensus$?
    \item Can our algorithms be generalized to allow deleting a certain number of outlier strings or columns in $\Scal$ such that the resulting substrings admit a solution (cf. \cite{BUL})?
    \item As the \FPT algorithms for those problems (even the classical algorithm from~\cite{GRA} for the Hamming distance) start with some given candidate, could they be studied from the perspective of learning-augmented algorithms?
\end{itemize}

\bibliographystyle{splncs04}
\bibliography{biblio}

\begin{thebibliography}{10}
\providecommand{\url}[1]{\texttt{#1}}
\providecommand{\urlprefix}{URL }
\providecommand{\doi}[1]{https://doi.org/#1}

\bibitem{AMIR2000247}
Amir, A., Aumann, Y., Landau, G.M., Lewenstein, M., Lewenstein, N.: Pattern
  matching with swaps. Journal of Algorithms  \textbf{37}(2),  247--266 (2000).
  \doi{https://doi.org/10.1006/jagm.2000.1120},
  \url{https://www.sciencedirect.com/science/article/pii/S0196677400911209}

\bibitem{AMI}
Amir, A., Paryenty, H., Roditty, L.: On the hardness of the consensus string
  problem. Information Processing Letters  \textbf{113},  371–374 (05 2013).
  \doi{10.1016/j.ipl.2013.02.016}

\bibitem{amir-3-strings}
{Amir, Amihood}, {Landau, Gad M.}, {Na, Joong Chae}, {Park, Heejin}, {Park,
  Kunsoo}, {Sim, Jeong Seop}: Consensus optimizing both distance sum and
  radius. In: Karlgren, J., Tarhio, J., Hyyr{\"o}, H. (eds.) String Processing
  and Information Retrieval. pp. 234--242. Springer Berlin Heidelberg, Berlin,
  Heidelberg (2009)

\bibitem{BasavarajuEtAl2018}
Basavaraju, M., Panolan, F., Rai, A., Ramanujan, M.S., Saurabh, S.: On the
  kernelization complexity of string problems. Theor. Comput. Sci.
  \textbf{730},  21--31 (2018)

\bibitem{betzler2011average}
Betzler, N., Guo, J., Komusiewicz, C., Niedermeier, R.: Average
  parameterization and partial kernelization for computing medians. Journal of
  Computer and System Sciences  \textbf{77}(4),  774--789 (2011)

\bibitem{BodlaenderEtAl2011}
Bodlaender, H.L., Thomass{\'{e}}, S., Yeo, A.: Kernel bounds for disjoint
  cycles and disjoint paths. Theor. Comput. Sci.  \textbf{412}(35),  4570--4578
  (2011)

\bibitem{BUL}
Bulteau, L., Schmid, M.L.: Consensus strings with small maximum distance and
  small distance sum. Algorithmica  \textbf{82}(5),  1378--1409 (2020)

\bibitem{CHEN2012164}
Chen, Z.Z., Ma, B., Wang, L.: A three-string approach to the closest string
  problem. Journal of Computer and System Sciences  \textbf{78}(1),  164--178
  (2012). \doi{10.1016/j.jcss.2011.01.003}, jCSS Knowledge Representation and
  Reasoning

\bibitem{cunha2025complexityalgorithmsswapmedian}
Cunha, L., Lopes, T., Mary, A.: Complexity and algorithms for swap median and
  relation to other consensus problems (2025),
  \url{https://arxiv.org/abs/2409.09734}

\bibitem{Dixon1996}
Dixon, J.D., Mortimer, B.: Permutation Groups. Springer New York (1996).
  \doi{10.1007/978-1-4612-0731-3}

\bibitem{Downey}
Downey, R.G., Fellows, M.R.: Parameterized Complexity. Springer Publishing
  Company, Incorporated (2012)

\bibitem{FominEtAl2019Book}
Fomin, F.V., Lokshtanov, D., Saurabh, S., Zehavi, M.: Kernelization: Theory of
  Parameterized Preprocessing. Cambridge University Press (2019)

\bibitem{FRA}
Frances, M., Litman, A.: On covering problems of codes. Theory of Computing
  Systems  \textbf{30},  113--119 (2007)

\bibitem{GRA}
Gramm, J., Niedermeier, R., Rossmanith, P.: Fixed-parameter algorithms for
  closest string and related problems. Algorithmica  \textbf{37},  25--42 (09
  2003). \doi{10.1007/s00453-003-1028-3}

\bibitem{damerau}
M., V.C., Rudy, S.Naga, D.: Fast and accurate spelling correction using trie
  and damerau-levenshtein distance bigram. TELKOMNIKA  \textbf{16}(2),
  827--833 (04 2018). \doi{10.12928/TELKOMNIKA.v16i2.6890}

\bibitem{10.1007/3-540-60044-2_50}
Muthukrishnan, S.: New results and open problems related to non-standard
  stringology. In: Galil, Z., Ukkonen, E. (eds.) Combinatorial Pattern
  Matching. pp. 298--317. Springer Berlin Heidelberg, Berlin, Heidelberg (1995)

\end{thebibliography}
\end{document}